\documentclass[journal, twoside]{IEEEtran}
\pdfoutput=1
\usepackage{graphics, amsmath}
\usepackage{epsfig}
\usepackage{subfigure}
\usepackage{xcolor}
\usepackage{float}
\usepackage[utf8]{inputenc}
\usepackage{ulem}

\let\labelindent\relax

\usepackage{enumitem}

\usepackage{amsthm}
\renewcommand{\QED}{{\QEDopen}}
\theoremstyle{definition}

\begin{document}
\pagestyle{empty}
\newcommand{\snr}{\textrm{snr}}
\newcommand{\sech}{\textrm{sech}}
\newcommand{\e}{\textrm{e}}
\newcommand{\snrstar}{{\textrm{snr}^{\star}}}
\newcommand{\SNR}{\textrm{SNR}}
\newcommand{\Bernoulli}{\textrm{Bernoulli}}

\newtheorem{lemma}{Lemma}
\newtheorem{corollary}{Corollary}
\newtheorem{theorem}{Theorem}
\newtheorem{definition}{Definition}
\newtheorem{algorithm}{Algorithm}
\newtheorem{remark}{Remark}
\newcommand{\bre}{\begin{equation}}
\newcommand{\ere}{\end{equation}}
\newcommand{\be}{{\bf {e}}}
\newcommand{\ee}\]
\newcommand{\bra}{\begin{eqnarray}}
\newcommand{\era}{\end{eqnarray}}
\newcommand{\bfg}{\begin{figure}[hbtp]}
\newcommand{\efg}{\end{figure}}
\newcommand{\bver}{\begin{verbatim}}
\newcommand{\ever}{\end{verbatim}}
\newcommand{\bit}{\begin{itemize}}
\newcommand{\eit}{\end{itemize}}
\newcommand{\ben}{\begin{enumerate}}
\newcommand{\een}{\end{enumerate}}
\newcommand{\ett}{\mbox{$\eta$} }
\newcommand{\coeff}[1]{\lfloor #1\rfloor}
\newcommand{\ceil}[1]{\lceil #1\rceil}
\newcommand{\floor}[1]{[ #1 ]}
\newcommand{\bfloor}[1]{\left[ #1 \right]}
\newcommand{\dgr}[1]{\mbox{$#1^{\circ}$}}
\newcommand{\cu}{\mbox{cosmic $\mu$ }}
\newcommand{\csa}[2]{\mbox{$\cos^2(#1 - #2)$}}
\newcommand{\csb}[2]{\mbox{$\cos2(#1 - #2)$}}
\newcommand{\balpha}{\mbox{\boldmath $\alpha$}}
\newcommand{\bbeta}{\mbox{\boldmath $\beta$}}
\newcommand{\blambda}{\mbox{\boldmath $\lambda$}}
\newcommand{\bkappa}{\mbox{\boldmath $\kappa$}}
\newcommand{\bXi}{\mbox{\boldmath $\Xi$}}
\newcommand{\brho}{\mbox{\boldmath $\rho$}}
\newcommand{\bchi}{\mbox{\bf x}}
\newcommand{\bal}{\mbox{\boldmath $\alpha$}_0}
\newcommand{\EE}{\mbox{\boldlarge E}}
\newcommand{\given}{\: | \:}
\newcommand{\Ker}{\mbox{Ker}\,}
\newcommand{\tildepss}{\tilde \epsilon_{s}}
\newcommand{\tildepssl}{\tilde \epsilon_{s_l}}
\newcommand{\epss}{\epsilon_s}
\newcommand{\epssl}{\epsilon_{s_l}}
\newcommand{\bphi}{{\mathbf \Phi}}
\newcommand{\bepsilon}{\mbox{\boldmath $\epsilon$}}
\newcommand{\bOmega}{{\mathbf \Omega}}
\newcommand{\bomega}{{\mathbf \omega}}
\newcommand{\bPhi}{{\mathbf \Phi}}
\newcommand{\bTheta}{{\mathbf \Theta}}
\newcommand{\btOmega}{\tilde{\mathbf \Omega}}
\newcommand{\bSigma}{{\mathbf \Sigma}}

\newcommand{\bsigma}{\mbox{\boldmath $\sigma$}}
\newcommand{\btsigma}{\mbox{\boldmath $\tilde\sigma$}}

\newcommand{\we}{\overrightarrow{e}}
\newcommand{\bA}{{\bf A}}
\newcommand{\bB}{{\bf B}}
\newcommand{\bC}{{\bf C}}
\newcommand{\bI}{{\bf I}}
\newcommand{\bg}{{\bf{g}}}
\newcommand{\bG}{{\bf{G}}}
\newcommand{\bE}{{\bf E}}
\newcommand{\bF}{{\bf F}}
\newcommand{\bm}{{\bf m}}
\newcommand{\bk}{{\bf k}}
\newcommand{\bof}{{\bf f}}
\newcommand{\rmf}{{\rm f}}
\newcommand{\norm}[1]{\|#1\|}
\newcommand{\oQ}{\overline Q}
\newcommand{\tQ}{\tilde Q}
\newcommand{\tD}{\tilde D}
\newcommand{\oD}{\bar D}
\newcommand{\oV}{\bar V}
\newcommand{\tP}{\tilde P}
\newcommand{\tN}{\tilde N}
\newcommand{\tA}{\tilde A}
\newcommand{\tM}{\tilde M}
\newcommand{\tX}{\tilde X}
\newcommand{\tY}{\tilde Y}
\newcommand{\tW}{\tilde W}
\newcommand{\tZ}{\tilde Z}
\newcommand{\tz}{\tilde z}
\newcommand{\tm}{\tilde m}
\newcommand{\tc}{\tilde c}
\newcommand{\tv}{\tilde v}
\newcommand{\tH}{\tilde H}
\newcommand{\btW}{\tilde {\bf W}}
\newcommand{\btA}{\tilde {\bf A}}
\newcommand{\btB}{\tilde {\bf B}}
\newcommand{\btV}{\tilde {\bf V}}
\newcommand{\btH}{\tilde {\bf H}}
\newcommand{\btZ}{\tilde {\bf Z}}
\newcommand{\btY}{\tilde {\bf Y}}
\newcommand{\btM}{\tilde {\bf M}}
\newcommand{\btX}{\tilde {\bf X}}
\newcommand{\btx}{\tilde {\bf x}}
\newcommand{\btc}{\tilde {\bf c}}

\newcommand{\ty}{\tilde y}
\newcommand{\btm}{\tilde {\bf m}}
\newcommand{\bFDPC}{\bF_{DPC}}
\newcommand{\tdet}{\tilde {\textrm{det}}}
\newcommand{\btT}{\tilde {\bf T}}
\newcommand{\btS}{\tilde {\bf S}}
\newcommand{\btn}{\tilde {\bf n}}
\newcommand{\btv}{\tilde {\bf v}}
\newcommand{\bty}{\tilde {\bf y}}
\newcommand{\btp}{\tilde {\bf p}}
\newcommand{\bta}{\tilde {\bf a}}
\newcommand{\PrEras}{\hat{\Pr}}
\newcommand{\ErasEpsilon}{\hat{\epsilon}}
\newcommand{\ErasCY}{\hat{\cY}}
\newcommand{\tab}{\ \ \ \ }
\newcommand{\bigtab}{\tab \tab \tab}
\newcommand{\bc}{{\bf c}}
\newcommand{\type}{{\mathrm type}}
\newcommand{\QEC}{{\mathrm{QEC}}}
\newcommand{\scnd}{{\mathrm {scnd}}}
\newcommand{\APP}{{\mathrm {APP}}}
\newcommand{\LLR}{{\mathrm {LLR}}}
\newcommand{\Cov}{{\mathrm{Cov}}}
\newcommand{\cov}{{\mathrm{cov}}}
\newcommand{\GF}{{\mathrm {GF}}}
\newcommand{\rank}{{\mathrm{rank\:}}}
\newcommand{\spn}{{\mathrm{span\:}}}
\newcommand{\E}{{\mathrm E}}
\newcommand{\bmu}{\mbox{\boldmath $\mu$}}
\newcommand{\bxi}{\mbox{\boldmath $\xi$}}
\newcommand{\cL}{{\cal L}}
\newcommand{\bN}{{\bf N}}
\newcommand{\cB}{{\cal B}}
\newcommand{\cH}{{\cal H}}
\newcommand{\cU}{{\cal U}}
\newcommand{\baa}{\begin{eqnarray*}}
\newcommand{\eaa}{\end{eqnarray*}}

\newcommand{\ds}{{\:d\bf s}}
\newcommand{\du}{{\:d\bf u}}
\newcommand{\dy}{{\:d\bf y}}
\newcommand{\dH}{{\:d\bf H}}
\newcommand{\dHyus}{\dH\dy\du\ds}
\newcommand{\diag}{\textrm{diag}}
\newcommand{\sign}{\textrm{sign}}
\newcommand{\abs}{\textrm{abs}}

\newcommand{\bs}{{\bf s}}
\newcommand{\ba}{{\bf a}}
\newcommand{\bb}{{\bf b}}
\newcommand{\bq}{{\bf q}}
\newcommand{\qstar}{{q^{\star}}}
\newcommand{\Qstar}{{Q^{\star}}}
\newcommand{\xstar}{{x^{\star}}}
\newcommand{\ystar}{{y^{\star}}}
\newcommand{\bp}{{\bf p}}
\newcommand{\bX}{{\bf X}}
\newcommand{\obX}{{\bar {\bf X}}}
\newcommand{\obx}{{\bar {\bf x}}}
\newcommand{\obm}{{\bar {\bf m}}}
\newcommand{\obY}{{\bar {\bf Y}}}
\newcommand{\oby}{{\bar {\bf y}}}
\newcommand{\obZ}{{\bar {\bf Z}}}
\newcommand{\obU}{{\bar {\bf U}}}
\newcommand{\obW}{{\bar {\bf W}}}
\newcommand{\obu}{{\bar {\bf u}}}
\newcommand{\obw}{{\bar {\bf w}}}
\newcommand{\obN}{{\bar {\bf N}}}
\newcommand{\obM}{{\bar {\bf M}}}
\newcommand{\obB}{{\bar {\bf B}}}
\newcommand{\oX}{{\bar {X}}}
\newcommand{\oY}{{\bar {Y}}}
\newcommand{\ou}{{\bar {u}}}
\newcommand{\ow}{{\bar {w}}}
\newcommand{\oR}{{\bar {R}}}
\newcommand{\oM}{{\bar {M}}}
\newcommand{\oB}{{\bar {B}}}
\newcommand{\bU}{{\bf U}}
\newcommand{\bW}{{\bf W}}
\newcommand{\bY}{{\bf Y}}
\newcommand{\bV}{{\bf V}}
\newcommand{\bZ}{{\bf Z}}
\newcommand{\bT}{{\bf T}}
\newcommand{\bS}{{\bf S}}
\newcommand{\bM}{{\bf M}}
\newcommand{\bh}{{\bf h}}
\newcommand{\bH}{{\bf H}}
\newcommand{\DFT}{{\mathrm{DFT}}}
\newcommand{\IDFT}{{\mathrm{IDFT}}}
\newcommand{\ldeg}{{\mathrm ldeg}}
\newcommand{\Real}{{\mathrm Re}}
\newcommand{\weight}{{\mathrm weight}}
\newcommand{\xor}{\oplus}
\newcommand{\bu}{{\bf u}}
\newcommand{\bv}{{\bf v}}
\newcommand{\bt}{{\bf t}}
\newcommand{\bd}{{\bf d}}
\newcommand{\bD}{{\bf D}}
\newcommand{\bw}{{\bf w}}
\newcommand{\bn}{{\bf n}}
\newcommand{\bx}{{\bf x}}
\newcommand{\br}{{\bf r}}
\newcommand{\by}{{\bf y}}
\newcommand{\bz}{{\bf z}}
\newcommand{\bone}{{\bf 1}}
\newcommand{\bzr}{{\bf 0}}
\newcommand{\cA}{{\cal A}}
\newcommand{\cP}{{\cal P}}
\newcommand{\cE}{{\cal E}}
\newcommand{\cF}{{\cal F}}
\newcommand{\cR}{{\cal R}}
\newcommand{\cS}{{\cal S}}
\newcommand{\cT}{{\cal T}}
\newcommand{\cX}{{\cal X}}
\newcommand{\cY}{{\cal Y}}
\newcommand{\oS}{\overline{S}}
\newcommand{\oP}{\overline{P}}
\newcommand{\hP}{\hat{P}}
\newcommand{\hR}{\hat{R}}
\newcommand{\tR}{\tilde{R}}
\newcommand{\hbF}{\hat{\bf F}}
\newcommand{\hF}{\hat{F}}
\newcommand{\hbU}{{\bf{\hat{U}}}}
\newcommand{\tbU}{{\bf{\tilde{U}}}}
\newcommand{\tbS}{{\bf{\tilde{S}}}}
\newcommand{\tTheta}{{\tilde{\Theta}}}
\newcommand{\tbTheta}{{\bf{\tilde{\Theta}}}}
\newcommand{\hbS}{{\bf{\hat{S}}}}
\newcommand{\hw}{\hat{w}}
\newcommand{\cC}{{\mathcal{C}}}
\newcommand{\cN}{{\mathcal{N}}}
\newcommand{\cG}{\mathcal{G}}
\newcommand{\hcC}{\hat{\mathcal{C}}}
\newcommand{\cD}{\mathcal{D}}
\newcommand{\hcD}{\hat{\mathcal{D}}}
\newcommand{\hD}{\hat{D}}
\newcommand{\tcD}{\tilde{\mathcal{D}}}
\newcommand{\tcC}{\tilde{\cal C}}
\newcommand{\tC}{\tilde{C}}
\newcommand{\tS}{\tilde{S}}
\newcommand{\hA}{\hat{A}}
\newcommand{\hB}{\hat{B}}
\newcommand{\hC}{\hat{C}}
\newcommand{\hc}{\hat{c}}
\newcommand{\hbr}{\hat{\bf r}}
\newcommand{\hbw}{\hat{\bf w}}
\newcommand{\hbc}{\hat{\bf c}}
\newcommand{\hby}{\hat{\bf y}}
\newcommand{\btz}{\tilde{\bf z}}
\newcommand{\hX}{\hat{X}}
\newcommand{\hY}{\hat{Y}}
\newcommand{\hy}{\hat{y}}
\newcommand{\hbx}{\hat{\bf x}}
\newcommand{\hbm}{\hat{\bf m}}
\newcommand{\hbX}{\hat{\bf X}}
\newcommand{\hbY}{\hat{\bf Y}}
\newcommand{\beginproof}{\noindent \textbf{Proof: }  }
\newcommand{\finproof}{\noindent $\Box$\\}
\newcommand{\mmse}{\textrm{mmse}}
\newcommand{\uncoded}{\textrm{uncoded}}
\newcommand{\ve}{\varepsilon}
\newcommand{\emptyline}{$\:\\ $}
\def\refeq#1{\: {\stackrel{ (#1)}{=}} \: }
\def\defined{\: {\stackrel{\scriptscriptstyle \Delta}{=}} \: }
\def\psdir#1{#1}
\def\MSE{{\rm MSE}}
\def\argmax{\mathop{\rm argmax}}
\def\CB{\mathop{\rm BS}}
\def\defined{\: {\stackrel{\scriptscriptstyle \Delta}{=}} \: }
\def\leqa{\buildrel \rm {\scriptscriptstyle (1)} \over \leq}
\def\leqb{\buildrel \rm {\scriptscriptstyle (2)} \over \leq}
\def\leqc{\buildrel \rm {\scriptscriptstyle (3)} \over \leq}
\def\eqa{\buildrel \rm {\scriptscriptstyle (1)} \over =}
\def\eqb{\buildrel \rm {\scriptscriptstyle (2)} \over =}
\def\eqc{\buildrel \rm {\scriptscriptstyle (3)} \over =}
\newfont{\boldlarge}{msbm10 scaled 1100}
\newcommand{\RR}{\mbox{\boldlarge R}}
\newcommand{\tr}{\mathrm{tr}}

\newcommand{\labeleq}[1]{\label{eq:#1}}
\newcommand{\refsec}[1]{(\ref{section:#1})}
\newcommand{\labelsec}[1]{\label{section:#1}}

\newcommand{\comment}[1]{}
\newcommand{\topc}[1]{$\stackrel{\circ}{\rm #1}$}
\newcommand{\ccc}[1]{$<$ \textbf{Remark: #1} $>$}
\newcommand{\Kavcic}{Kav\u{c}i\'{c}}

\def\etal{$et \,\,al.$}

\renewcommand{\thesection}{\Roman{section}}

\newlength{\tmpbigbar}

\newcommand{\bigbar}[1]{
 \setlength{\tmpbigbar}{\unitlength}
 \settowidth{\unitlength}{\mbox{$#1$}}
  \stackrel{\barpic}{#1}
 \setlength{\unitlength}{\tmpbigbar}
}
\newcommand{\barpic}{\begin{picture}(1,0.01)(0,0)
\put(0.15,0){\line(12,0){0.7}}
\end{picture}
}

\newcommand\addabove[2]{\: {\stackrel{\scriptscriptstyle (#2)}{#1}} \: }

\newcommand{\myfigure}[3]{\begin{figure}[htp]
\begin{center}
\epsfig{#1}
\end{center}
\caption{#2}\label{#3}
\end{figure}}

\newcommand{\myfigurestar}[3]
{\begin{figure*}
    \begin{center}
    \epsfig{#1}
    \caption{#2}
    \label{#3}
    \end{center}
\end{figure*}}

\newcommand{\TwoQuad}{\quad\quad}
\newcommand{\ThreeQuad}{\quad\quad\quad}
\newcommand{\FourQuad}{\quad\quad\quad\quad}
\newcommand{\FiveQuad}{\quad\quad\quad\quad\quad}
\newcommand{\SixQuad}{\quad\quad\quad\quad\quad\quad}
\newcommand{\SevenQuad}{\quad\quad\quad\quad\quad\quad\quad}
\newcommand{\AlighThisEquation}{&&\hspace{-0.5cm}}



\title{Deep Learning for Decoding of Linear Codes -\\* A Syndrome-Based Approach}

\newcommand\blfootnote[1]{%
  \begingroup
  \renewcommand\thefootnote{}\footnote{#1}%
  \addtocounter{footnote}{-1}%
  \endgroup
}
\author{
\authorblockN{Amir~Bennatan\textsuperscript{*}\thanks{* Both authors contributed equally to this work.}, Yoni Choukroun\textsuperscript{*} and Pavel Kisilev}\\
\authorblockA{
Huawei Technologies Co.,\\
Email: \{amir.bennatan, yoni.choukroun, pavel.kisilev\}@huawei.com} }

\maketitle
\pagestyle{empty}
\thispagestyle{empty}

\begin{abstract}
We present a novel framework for applying deep neural networks (DNN) to soft decoding of linear codes at arbitrary block lengths.  Unlike other approaches, our framework allows unconstrained DNN design, enabling the free application of powerful designs that were developed in other contexts.  Our method is robust to overfitting that inhibits many competing methods, which follows from the exponentially large number of codewords required for their training.  We achieve this by transforming the channel output before feeding it to the network, extracting only the syndrome of the hard decisions and the channel output reliabilities.   We prove analytically that this approach does not  involve any intrinsic performance penalty, and guarantees the generalization of performance obtained during training.  Our best results are obtained using a recurrent neural network (RNN) architecture combined with simple preprocessing by permutation.   We provide simulation results that demonstrate performance that sometimes approaches that of the ordered statistics decoding (OSD) algorithm.
\end{abstract}
\section{Introduction} \label{sec:Introduction}

Interest in applying neural networks to decoding has existed since the 1980's~\cite{Old_1,Old_4,Old_6}.  These early works, however, did not have a substantial impact on the field due to the limitations of the networks that were available at the time.
More recently, deep neural networks were studied in~\cite{Hoydis_1,Hoydis_2,Nachmani_1}.

A major challenge facing applications of deep networks to decoding is the avoidance of overfitting the codewords encountered during training. Specifically, training data is typically produced by randomly selecting codewords and simulating the channel transitions.  Due to the large number of codewords (exponential in the block length), it is impossible to account for even a small fraction of them during training, leading to poor generalization of the network to new codewords.  This issue was a major obstacle in~\cite{Hoydis_1,Hoydis_2}, constraining their networks to very short block lengths.

Nachmani~\etal~\cite{Nachmani_1},~\cite{Nachmani_3} proposed a deep learning framework which is modeled on the LDPC belief propagation (BP) decoder, and is robust to overfitting.  A drawback of their design, however, is that to preserve symmetry, the design is constrained to closely mimic the message-passing structure of BP.  Specifically, the connections between neurons are controlled to resemble BP's underlying Tanner graph, as are the activations at neurons.  This severely limits the  freedom available to the neural networks design, and precludes the application of powerful architectures that have emerged in recent years~\cite{DeepLearningTextBook}.

In this paper, we present a method which overcomes this drawback while maintaining the resilience to overfitting of~\cite{Nachmani_1}.  Our framework allows unconstrained neural network design, paving the way to the application of powerful neural network designs that have emerged in recent years.  Central to our approach is a preprocessing step, which extracts from the channel output only the reliabilities (absolute values), and the syndrome of its hard decisions, and feeds them into the neural network. The network's output is later combined with the channel output to produce an estimate of the transmitted codeword.

Decoding methods that focus on the syndrome are well known in literature on algebraic decoding (see e.g.~\cite{CostelloBook}[Sec. 3.2]).  The approach decouples the estimation of the channel noise from that of the transmitted codeword.  In our context of deep learning, its potential lies in the elimination of the need to simulate codewords during training, thus overcoming the overfitting problem.  A few early works that used shallow neural networks have employed syndromes (e.g.~\cite{Old_1}).  However, these works did not discuss its potential in terms of overfitting, presumably because the problem was not as acute in their relatively simple networks.  Importantly, their approach does not apply in cases where the channel includes reliabilities (mentioned above).

Our approach in this paper extends syndrome decoding to include channel reliabilities, and applies it to overcome the overfitting issue mentioned above.  We provide a rigorous analysis which proves that our framework incurs no loss in optimality, in terms of bit error rate (BER) and mean square error (MSE).  Our analysis utilizes some techniques developed by Burshtein~\etal~\cite{Sason_2}, Wiechman and Sason~\cite{Sason_1} and Richardson and Urbanke~\cite{UrbankeBook}[Sec.~4.11].

Building on our analysis, we propose two deep neural network architectures for decoding of linear codes.  The first is a vanilla multilayer network, and the second a more-elaborate recurrent neural network (RNN) based architecture.  We also develop a preprocessing technique (beyond the above-mentioned computation of the syndrome and reliabilities), which applies a permutation (an automorphism) to the decoder input to facilitate the operation of the neural network.  This technique builds on ideas  by Fossorier~\etal~\cite{OSD_1,OSD_2} and Dimnik and Be'ery~\cite{Beery} but does {\it not} involve list decoding.  Finally, we provide simulation results for decoding of BCH codes which for the case of BCH(63,45), demonstrate performance approaches that of the ordered statistics algorithm (OSD) of~\cite{OSD_1,OSD_2}.

Summarizing, our main contributions are:
\begin{enumerate}
\item A novel deep neural network training framework for decoding, which is robust to overitting.  It is based on the extension of syndrome decoding that is described below.
\item An extension of syndrome decoding which accounts for reliabilities.   As with legacy syndrome decoding, we define a generic framework, leaving room for a noise-estimation algorithm which is specified separately. We provide analysis that proves that the framework involves no loss of optimality, and that regardless of the noise-estimation algorithm, performance is invariant to the transmitted codeword.   
\item Two neural network designs for the noise-estimation algorithm, including an elaborate RNN-based architecture.
\item A simple preprocessing technique that applies permutations to boost the decoder's performance.
\end{enumerate}

Our work is organized as follows.  In Sec.~\ref{sec:Notation} we introduce some notations and in Sec~\ref{sec:nn} we provide some bakground on neural networks.  In Sec.~\ref{sec:Theorem_Main} we describe our syndrome-based framework and provide an analysis of it.  In Sec.~\ref{sec:Architecture} we discuss  deep neural network architectures as well as preprocessing by permutation.  In Sec.~\ref{sec:simulation_results} we present simulation results.  Sec.~\ref{sec:Conclusion} concludes the paper.

\section{Notations}\label{sec:Notation}

\def\bipolar{\textrm{bipolar}}
\def\bin{\textrm{bin}}
\def\MAP{\textrm{MAP}}

We will often use the superscripts $b$ and $s$ (e.g., $x^b$ and $x^s$) to denote binary values (i.e, over the alphabet $\{0,1\}$) and bipolar values (i.e., over  $\{\pm1\}$), respectively.  We define the mapping from the binary to the bipolar alphabet, denoted $ \bipolar(\cdot)$, by $0 \rightarrow 1, 1\rightarrow -1$ and let $\bin(\cdot)$ denote the inverse mapping.
Note that the following identity holds:
\begin{eqnarray}\label{eq:3}
\bin(x^s\cdot y^s) = \bin(x^s)\xor\bin(y^s), \tab\forall x^s,y^s\in\{\pm 1\}
\end{eqnarray}
where $\xor$ denotes XOR.   $\sign(x)$ and $|x|$ denote the sign and absolute value of any real-valued $x$, respectively,  In our analysis below, when applied to a vector, the operations $\bipolar(\cdot)$, $\bin(\cdot)$, $\sign(\cdot) $ and $|\cdot|$ are assumed to be applied independently to each of the vector's components.

\section{Brief Overview of Deep Neural Networks}\label{sec:nn}

We now briefly describe the essentials of neural networks.  Our discussion is in no way comprehensive, and there are many variations on the simple setup we describe.  For an elaborate discussion, see e.g.~\cite{DeepLearningTextBook}.

\myfigure{file=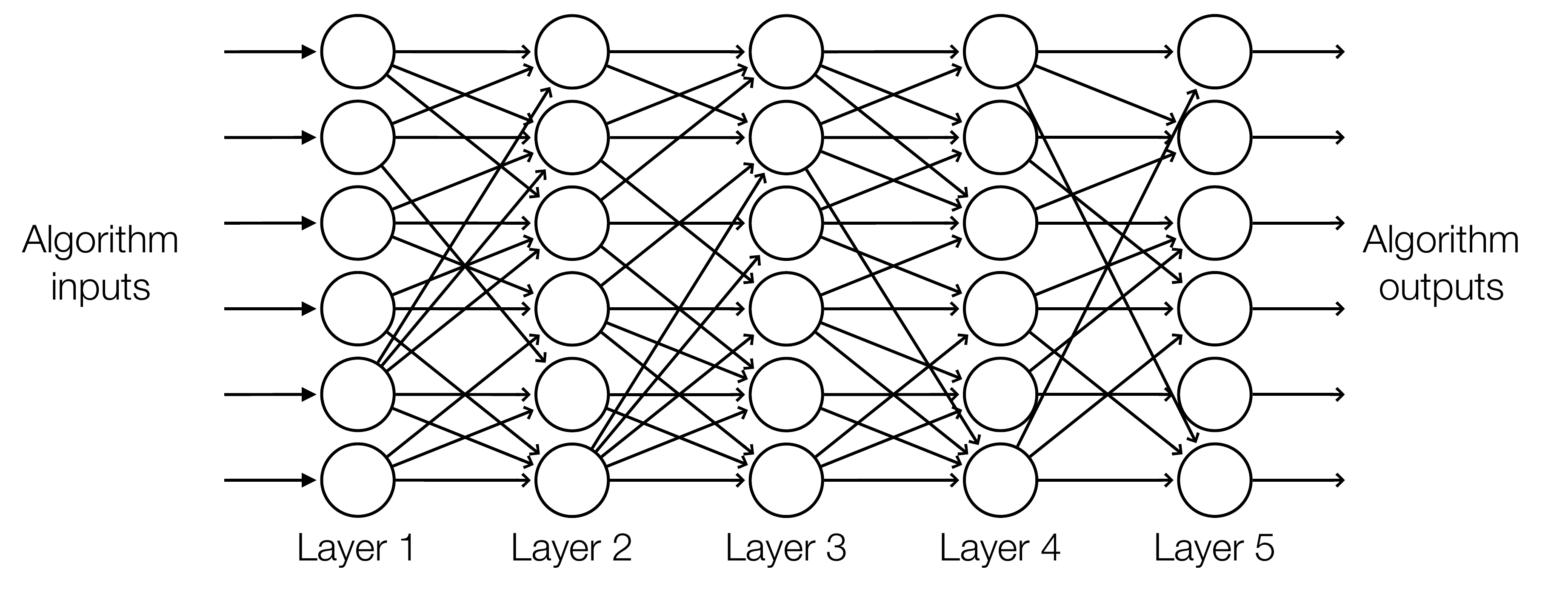, width = 8cm}{Illustration of a multi-layer neural network.}{fig:6}

Fig.~\ref{fig:6} depicts a simple multilayered neural network.  The network is a directed graph whose structure is the blueprint for an associated computational algorithm. 
The nodes are called {\it neurons} and each performs (i.e., is associated with) a simple computation on inputs to produce a single output.   The algorithm inputs are fed into the first-layer neurons, whose outputs are fed as inputs into the next layer and so forth.  Finally, the outputs of the last layer neurons become the algorithm output.  {\it Deep} networks are simply neural networks with many layers.

The inputs and output at each neuron are real-valued numbers. To compute its output, each neuron first computes an affine function of its input (a weighted sum plus a {\it bias}).  It then applies a predefined {\it activation} function, which is typically nonlinear (e.g., sigmoid or hyperbolic tangent), to render the output.

The power of neural networks lies in their configurability.  Specifically, the weights and biases at each neuron are parameters which can be tailored to produce a diverse range of computations.  Typically, the network is configured by a training procedure.  This procedure relies on a sample dataset, consisting of inputs, and in the case of {\it supervised} training, of desired outputs (known as {\it labels}).  There are several training paradigms available, most of which are variations of gradient descent.

Overfitting occurs typically when the training dataset is insufficiently large or diverse to be representative of all valid network inputs.  The resulting network does not generalize well to inputs that were not encountered in the training set.


\section{The Proposed Syndrome-Based Framework and its Analysis} \label{sec:Theorem_Main}
\subsection{Framework Definition} \label{sec:1}

We begin by briefly discussing the encoder.  We assume standard transmission that uses a linear code $C$.  We let $\bm \in \{0,1\}^K$ denote the input message, which is mapped to a codeword $\bx^b \in \{0,1\}^N$ (recall that the superscript $b$ denotes binary vectors).   We assume that $\bx^b$ is mapped to a bipolar vector $\bx$  using the  mapping defined in Sec.~\ref{sec:Notation}, and $\bx$ is transmitted over the channel.  We let $\by$ denote the channel output.  For convenience, we define $\bm$, $\bx^b$, $\bx$ and $\by$ to be column vectors. 

\myfigure{file=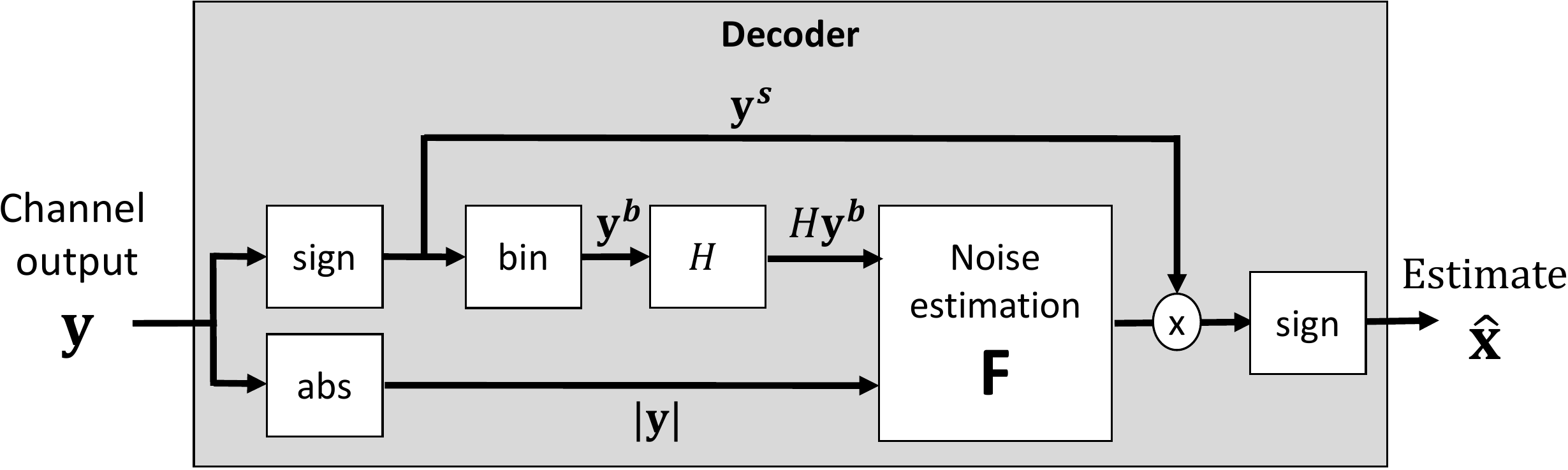, width = 9cm}{Decoder framework.  The  final $\sign$ operation is omitted when the system is designed to produce soft decisions.  In Sec.~\ref{sec:Architecture}, we will implement $\bF$ using a neural network.}{fig:1}
 
 Our decoder framework is depicted in Figure~\ref{fig:1}.  The decoder's main component is $\bF$, whose role is to estimate the channel noise, and which we will later (Section~\ref{sec:Architecture}) implement using a deep neural network.  The discussion in this section, however, applies to arbitrary $\bF$.
The inputs to $\bF$ are the absolute value $|\by|$ and the {\it syndrome} $H\by^b$, where $H$ is a parity check matrix of the code $C$, $\by^b$ is a vector of hard decisions, $\by^b = \bin(\sign(\by))$ where $\bin(\cdot)$ is simply the inverse of the above defined bipolar mapping.  The multiplication by $H$ is modulo-2.  The output of $\bF$ is multiplied (componentwise) by $\by^s$, which is the sign of $\by$.  Finally, when interested in hard decisions, we take the sign of the results and define this to be the estimate $\hbx$.  This final hard decision step (which is depicted in Fig.~\ref{fig:1}) is omitted when the system is required to produce soft decisions.

\subsection{Binary-Input Symmetric-Output (BISO) Channels}

Our analysis in the following section applies to a broad class of channels known as binary-input symmetric-output (BISO) channels.  This class includes binary-input AWGN channels, binary-symmetric channels (BSCs) and many others.  Our definition below follows  Richarson~\etal~\cite{Urbanke_Message_Passing}[Definition~1].  
\begin{definition}
Consider a memoryless channel with input alphabet $\{\pm1\}$.  The channel is BISO if its transition probability function satisfies:
\begin{eqnarray}
\Pr[Y=y\given X=1] = \Pr[Y=-y\given X=-1], 
\label{eq:6}
\end{eqnarray}
for al $y$ the channel output alphabet, where $X$ and $Y$ denote the random channel input and output (respectively).
\end{definition}
\vspace{0.2cm}
An important feature of BISO channels is that their random transitions can be modeled by~\cite{Urbanke_Message_Passing}[proof of Lemma~1],
\begin{eqnarray}
Y = X\cdot \tZ,\label{eq:5}
\end{eqnarray}
where~$\tZ$ is random noise which is independent of the transmitted $X$.  The tilde in $\tZ$ serves to indicate that this is an equivalent statistical model, which might differ from the true physical one.  To prove~\eqref{eq:5}, we simply define $\tZ$ to be a random variable distributed as~$\Pr[Y\given X=1]$ and independent of $X$.  The validity now follows from~\eqref{eq:6}.

\subsection{Analysis}\label{sec:analysis}
We now show that the decoder framework involves no penalty in performance in terms of metrics mean-squared-error (MSE) or bit error rate (BER).  That is, the decoder can be designed to achieve any $\it{one}$ of them.  Importantly, it addresses the overfitting problem that was described in Sec.~\ref{sec:Introduction}.
\begin{theorem} \label{thrm:Main} The follwing holds with respect to the framework of Sec.~\ref{sec:1}, assuming communication over a BISO channel:
\begin{enumerate}
\item The framework incurs no intrinsic loss of optimality, in the sense that with an appropriately designed $\bF$, the decoder can achieve maximum {\it a-posteriori} (MAP) or minimum MSE (MMSE) decoding. \label{item:1}
\item For any choice of $\bF$, the decoder's BER and MSE, conditioned on transmission of any codeword $\bx$, are both invariant to $\bx$.\label{item:2}
\end{enumerate}
\end{theorem}
\vspace{0.2cm}
We provide an outline of the proof here, and defer the details to 
Appendix~\ref{apdx:proof}.

\begin{proof} [Outline]  In Part~1 we neglect implementation concerns, and focus on realizations of $\bF$ that try to optimally estimate the multiplicative noise $\btz$ (see~\eqref{eq:5}).  By~\eqref{eq:5}, given $\by$, such estimation is equivalent to estimation of $\bx$. We argue that the pair $|\by|$ and $H\by^b$ is a sufficient statistic for estimation of $\btz$.  To see this, observe that $\by$ is equivalent to the pair $|\by|$ and $\sign(\by)$.  The latter term, in turn, is equivalent to the pair $H\by^b$ and $A\by^b$, where $\by^b$ is as defined in Sec.~\ref{sec:Notation} and $A$ is a pseudo-inverse of the code's generator matrix $G$.  By~\eqref{eq:3} and~\eqref{eq:5}, $\by^b = \bx^b\xor \btz^b$ and
so $A\by^b$ is the sum of the transmitted message $\bm=A\bx^b$ and $A\btz$ (the projection of $\btz$ onto the code subspace).  We argue that $A\by^b$ is independent of the noise and thus irrelevant to its estimation.  This follows because $\bm$ is independent of $\btz$, and we assume it to be uniformly distributed within the message space $\{0,1\}^K$.

\newcommand{\bff}{{\bf f}}

To prove Part~2 of the theorem, we allow $\bF$ to be arbitrary and show that the decoder's output can be modeled as $\bff(\btz)\cdot\bx$  for some vector-valued function $\bff(\cdot)$.  Thus, its relationship with $\bx$ (which determines the BER and MSE) depends on the noise $\btz$ alone.  To see why this holds, first observe that the inputs to $\bF$ (and consequently, its outputs) are dependent on the noise $\btz$ only.  This follows because the syndrome $H\by^b$ equals $H\btz^b$.  This in turns follows from the relation $\by^b = \bx^b\xor\btz^b$ and the fact that $\bx^b$ is a codeword, and so $H\bx^b={\bf 0}$.  By~\eqref{eq:5} and the bipolarity of $\bx$, the absolute value $|\by|$ equals $|\btz|$.  It now follows that the output of $\bF$ is  dependent on $\btz$ alone.  Multiplication by $\sign(\by)$ (see~Fig. \ref{fig:1})  is equivalent to multiplying by $\sign(\btz)\cdot\bx$ (by~\eqref{eq:5}) and the result follows.
\end{proof}

\section{Implementation using Deep Neural Networks}\label{sec:Architecture}

Theorem~\ref{thrm:Main} proves that our framework does not intrinsically involve a loss of optimality.  To realize its potential, we propose efficient implementations of the function $\bF$.  In this section, we discuss deep neural network implementations as well as a simple preprocessing technique that enables the networks to achieve improved performance.

\subsection{Deep Neural Network Architectures} \label{sec:arch}

We consider the following two architectures:

\begin{enumerate}
  \setlength{\itemsep}{1pt}
  \setlength{\parskip}{1pt}
  
\item {\bf Vanilla Multi-Layer:}
With this architecture, the neural network contains an array of fully-connected layers as illustrated in Fig.~\ref{fig:66}.  
It closely resembles simple designs~\cite{DeepLearningTextBook}  with the exception that we feed the network inputs into each of the layers, in addition to the output of the previous layer (this idea is borrowed from the belief propagation algorithm). 

The network includes 11 fully-connected layers, roughly equivalent to the 5 LDPC belief-propagation iterations as in~\cite{Nachmani_1}.    We use rectified linear unit (ReLU) nonlinear activation functions~\cite{DeepLearningTextBook}.
Each of the first 10 layers consists of the same number of nodes ($6N$ for block length $N=64$ and $15N$ for block length $N=127$). The final fully-connected layer has $N$ nodes and produces the network output, using a hyperbolic tangent arctivation function.

\myfigure{file=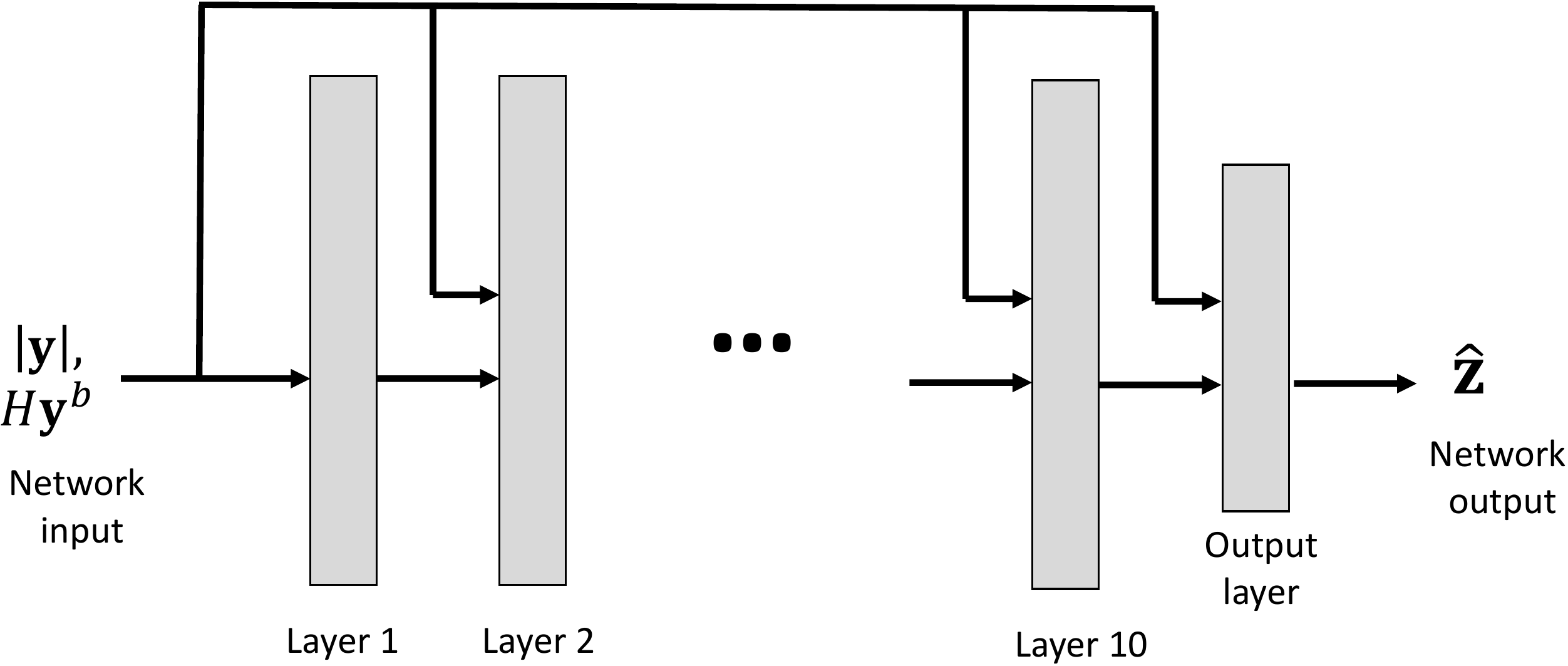, width = 8cm}{An instance of the Vanilla multi-layer architecture.}{fig:66}

\item{\bf Recurrent Neural Network (RNN):} 
With this architecture, we build a deep RNN by stacking multiple recurrent hidden states on top of each other~\cite{graves2013speech} as illustrated in Fig.~\ref{fig:3}.  RNNs realize a design which is equivalent to a multi-layer architecture by maintaining a network {\it hidden state} (memory) which is updated via feedback connections.  Note that from a practical perspective, this renders the network more memory-efficient. In many applications, this structure is useful to enable temporal processing, and stacking RNNs as in Fig.~\ref{fig:3} enables operation at different timescales.  In our setting, this interpretation does not apply, but we nonetheless continue to refer to RNN layers as ``time steps.''  Stacking multiple RNNs produces an effect that is similar to deepening the network.

We use Gated Recurrent Unit (GRU)~\cite{cho2014learning} cells which have shown peformance similar to well-known long short-term memory (LSTM) cells~\cite{hochreiter1997long}, but have less parameters due to the lack of a reset gate, making them a faster inference alternative.
We use the hyperbolic tangent nonlinear activation functions, the networks posses 4 RNN stacks (levels), the hidden state size is set to $5N$ and the RNN performs 5 time steps. 
\end{enumerate}
\myfigure{file=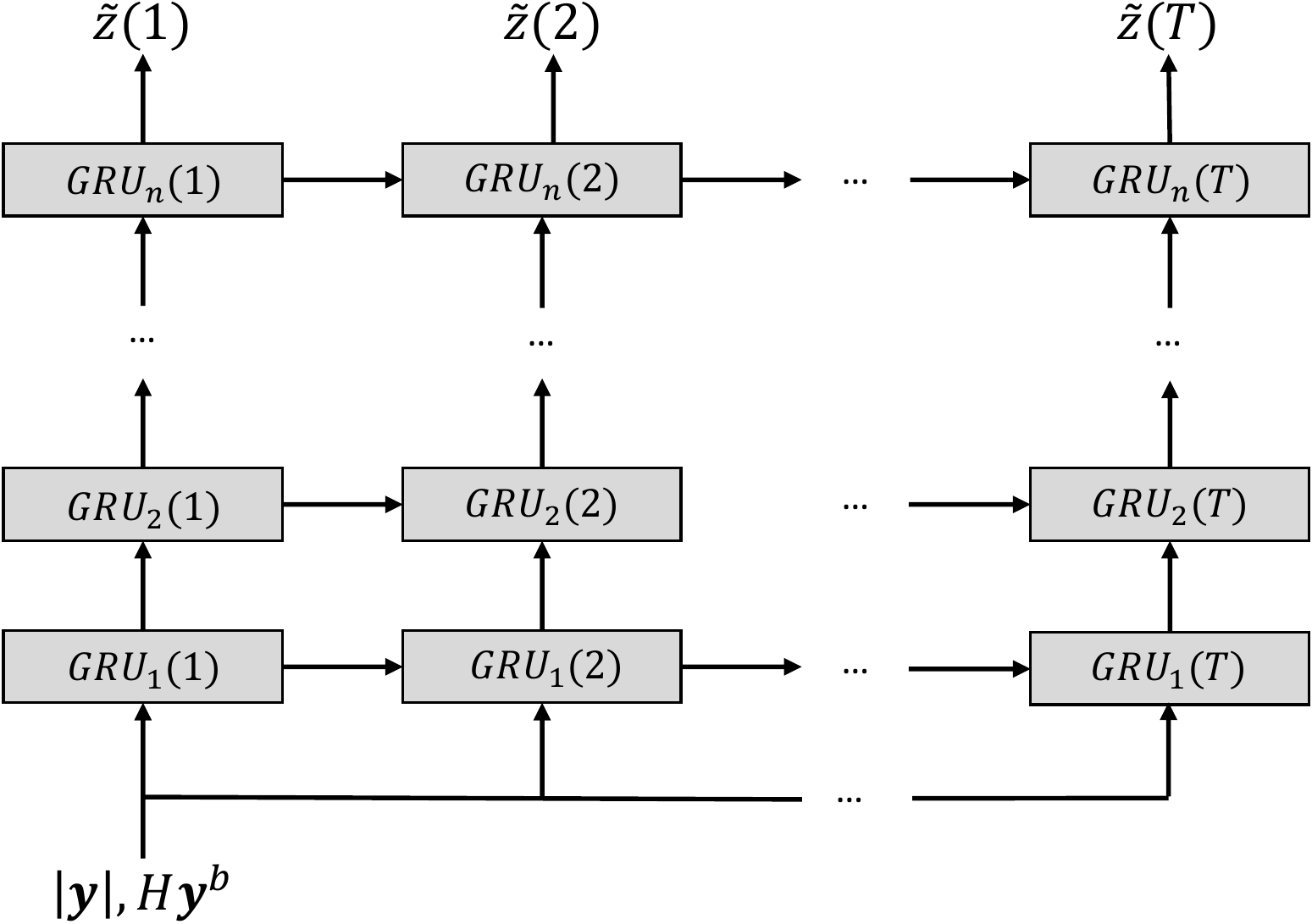, width = 7cm}{Our implementation of stacked RNN layers where $GRU_{i}(t)$ represents the $i$-th level cell at RNN time step (layer) $t$. }{fig:3}

To train the networks, we simulate transmission of the all-one codeword (assuming the bipolar alphabet, $\{\pm1\}$).  
We also simulate the multiplicative noise $\btz$, which in the case of an AWGN channel, is distributed as a mean-$1$ Gaussian random variable.  In our training for Sec.~\ref{sec:simulation_results}, we  set $E_b/N_0$ to 4 dB.  This was selected arbitrarily, and could potentially be improved. We use Google's TensorFlow library and the Adam optimizer~\cite{Adam}. Testing proceeds in the same lines, except that we use randomly generated codewords rather than the all-one codeword.  With each training batch, we generate a new set of noise samples.  While this procedure produced our best results, an alternative approach which fixes the training noise and uses other techniques (e.g., dropout) to overcome overfitting the noise, is worth exploring.

With the RNN architecture, the network produces multiple outputs (at each time-step) and we use the following loss function:
\begin{equation*} 
L=\frac{1}{N}\sum_{t=1}^{T} \sum_{i=1}^N\gamma^{T-t}H_\textrm{CE}(\tilde{z}_i^s,\hat{z}_i^s(t)),
\end{equation*}
where $H_\textrm{CE}$ is the cross-entropy function, $\tilde{z}^s_i$ is the sign of component $i$ of the multiplicative noise and  where  $\hat{z}_i^s(t)$ the network output corresponding to codebit $i$ at RNN time step (layer) $t$.  $\gamma < 1$ a discount factor (in our simulations, we used $\gamma = 0.5$).  The loss for the vanilla architecture is a degenerate version of the RNN one, with time steps and discount factors removed.


\subsection{Preprocessing by Permutation} \label{sec:preproc}
The performance of the implementations described above can further be improved by applying simple preprocessing  and postprocessing steps at the input and output of the decoder (respectively).  Our approach is depicted in Fig.~\ref{fig:7}. Preprocessing involves applying a permutation to the components of the channel output vector, and postprocessing applies the inverse permutation to the decoder output.  The approach draws on ideas from 
Fossorier~\etal~\cite{OSD_1,OSD_2} and Dimnik and Be'ery~\cite{Beery}.   Note that the approach deviates from these works in that it does {\it not} involve computing a list of vectors.  

\myfigure{file=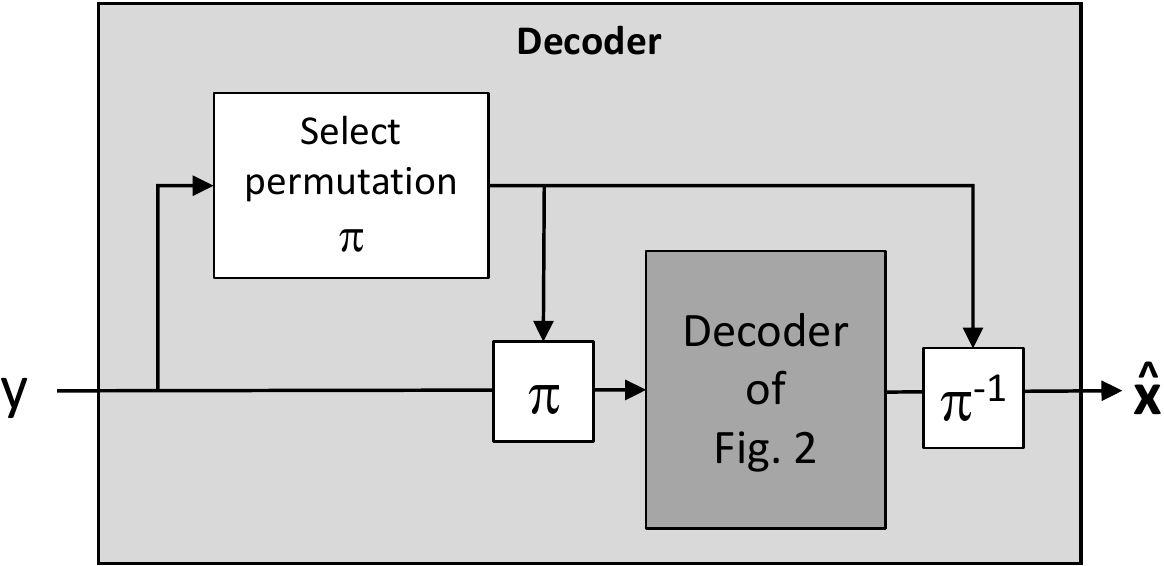, width = 8cm}{Decoder with perprocessing and postprocessing components.  Note that in Appendix~\ref{apdx:BCH_Perm} we show that this decoder has an equivalent representation which is a special case of Fig.~\ref{fig:1}.}{fig:7}

Similar to~\cite{OSD_1,OSD_2}, our decoder selects the preprocessing permutation so as to maximize the sum of the adjusted reliabilities of the first $K$ components of the permuted channel output vector. 
Borrowing an idea from~\cite{Beery}, however, we confine our permutations to subsets of the code's automorphism group.   We assume that the parity check matrix, by which the syndrome in Fig.~\ref{fig:7} is computed, is arranged so that  the last $N-K$ columns are diagonal and correspond to parity bits of the code.  

We define the adjusted reliability of a channel component $y_i$, denoted $R(y_i)$ by,
\begin{eqnarray}\label{eq:R}
R(y) \defined I(X;Y\given \:\:|Y|=|y|).
\end{eqnarray}
That is, $R(y)$ equals the mutual information of random variables $X$ and $Y$, denoting the channel input and output, respectively, conditioned on the event that the absolute value $|Y|$ equals $|y|$.  $X$ is uniformly distributed in $\{\pm1\}$ and $Y$ is related to it via the channel transition probabilities.  With this definition, our permutation selection criterion is equivalent to concentrating as much as possible of the channel capacity within the first $K$ channel output components.  

Unlike~\cite{OSD_2}, we borrow an idea from~\cite{Beery} and restrict the set of allowed permutations to the code's {\it automorphism group}~\cite{Sloane}.  Permutations in this group have the property that the permuted version of any codeword is guaranteed to be a codeword as well.  In our context, confinement to such permutations ensures that the decoder input (the permuted channel output) continues to obey the communication model, namely being a noisy valid codeword.  The decoder can thus continue to rely on the code's structure to decode.  

When compared to framework of Fig.~\ref{fig:1}, the decoder of Fig.~\ref{fig:7} has the added benefit of knowing that the $K$ first channel outputs are consistently more reliable than the remaining components.   That is, the input exhibits additional structure that the neural network can rely on.

In Appendix~\ref{apdx:BCH_Perm} we discuss permutations for BCH codes like  those  we will use in Sec~\ref{sec:simulation_results} below, as well as efficient methods for computing the optimal permutation.  We also discuss formal aspects related to applying the analysis of Sec.~\ref{sec:Theorem_Main} to the framework of Fig.~\ref{fig:7}.  
Note that to achieve good results, the added steps of Fig.~\ref{fig:7} need to be included during training of the neural network.  

In Sec.~\ref{sec:simulation_results} we present simulation results for BCH(127,64) codes with and without the above preprocessing method, demonstrating the effectiveness of this approach.  Note that with the shorter block length BCH(63,45) codes, preprocessing was not necessary, and we obtained performance that approaches the ordered statistics algorithm even without it.


\section{Simulation Results}\label{sec:simulation_results}

Fig.~\ref{fig:2} presents simulation results for communication with the BCH(63,45) code over an AWGN channel.  We simulated our two architectures, namely syndrome-based vanilla and stacked-RNN. Note that in this case, we did {\it not} simulate permutations (Sec.~\ref{sec:preproc}).  Also plotted are results for the best method of Nachmani~\etal~\cite{Nachmani_1},\cite{Nachmani_3}, the belief propagation (BP) algorithm, and for the ordered statistics decoding (OSD) algorithm~\cite{OSD_2} of order 2 (for this algorithm we simulated $10^4$ codewords for each $E_b/N_0$ point).    As can be seen from the results, both our architectures substantially outperform the BP algorithm.  Our stacked-RNN architecture, like that of~\cite{Nachmani_3}, approaches the OSD algorithm very closely.  

Fig.~\ref{fig:4} presents results for the BCH(127,64) code.  We simulated our syndrome-based stacked RNN method, with and without preprocessing.   As can be seen, the preprocessing step renders as a substantial benefit.  Also plotted are results for the BP and the OSD algorithms as well as the best results of~\cite{Nachmani_1}\footnote{Their paper~\cite{Nachmani_3} does not include results for this case.}.  Both our methods outperform the BP algorithm as well as the algorithm of~\cite{Nachmani_1}.   However, a gap remains to the OSD algorithm\footnote{Note that for $E_b/N_0$ of 4 dB or higher, we encountered no OSD errors for BCH(127,64), in our simulations.}, which widens with $E_b/N_0$.

With respect to the number of codewords simulated, with our algorithms, we simulated $10^5$ codewords for each $E_b/N_0$ point.  With the OSD algorithm, we simulated $10^4$ codewords for each $E_b/N_0$ point.  With the BP algorithm we simulated $10^3$ for each point.  

The analysis of~\cite{Nachmani_1} also includes an mRRD framework, into which their algorithm can be plugged to obtain superior performance.  While our algorithms can similarly be plugged into that framework, our interest in this paper is in methods whose primary components are neural networks.

\myfigure{file=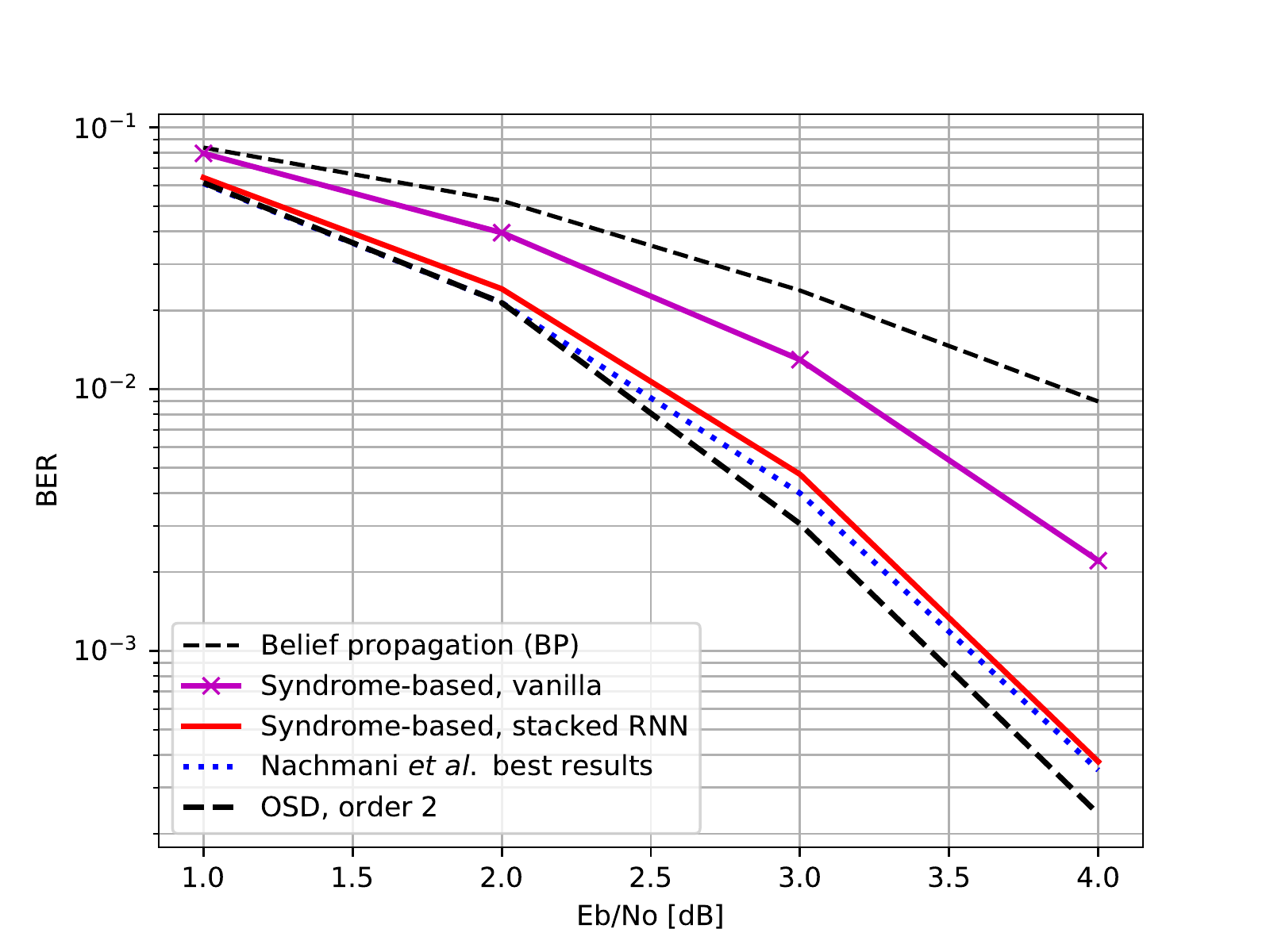, width = 9.5cm}{BER results for a BCH (63,45) code.}{fig:2}

\myfigure{file=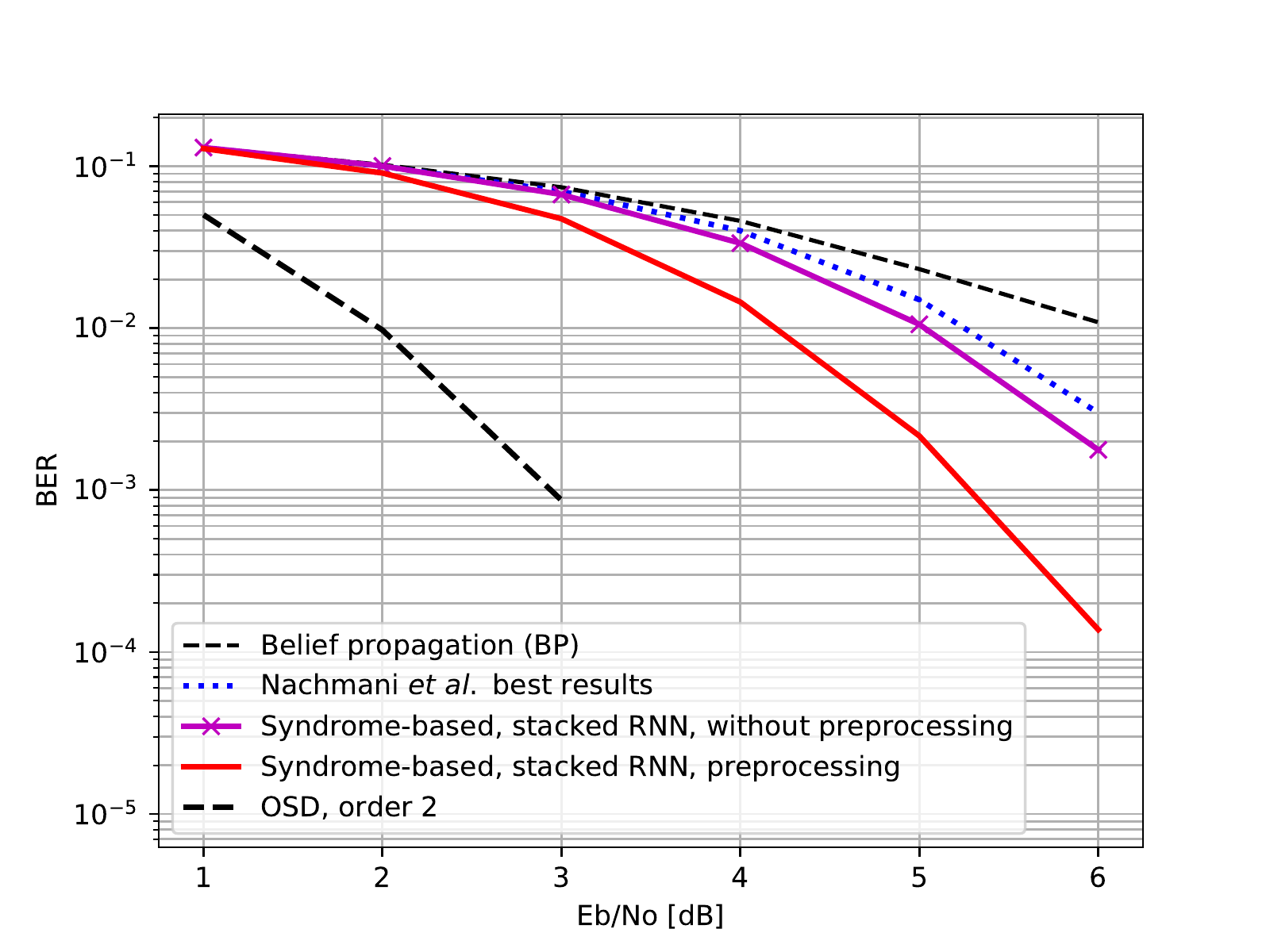, width = 9.5cm}{BER results for a BCH (127,64) code.}{fig:4}


\section{Conclusion}\label{sec:Conclusion}

Our work in this paper presents a promising new framework for the  application of deep learning to error correction coding.  An important benefit of our design is the elimination of the problem of overfitting to the training codeword set, which was experienced by~\cite{Hoydis_1,Hoydis_2}.  We achieve this by using  syndrome decoding, and by extending it to account for soft channel reliabilities.  Our approach enables the neural network to focus on the estimating the noise alone, rather than the transmitted codeword.    

It is interesting to compare our framework to that of Nachmani~\etal~\cite{Nachmani_1}. While their approach also resolves the overfitting problem and achieves impressive simulation results, it is heavily constrained to follow the structure of the LDPC belief-propagation decoding algorithm.  By contrast, our framework allows the unconstrained design of the neural network, and our architectures are free to draw from the rich experience that has emerged in recent years on neural network design.

Our simulation results demonstrate that our framework can be applied to achieve strong performance that approaches OSD.  Further research will examine additional  neural network architectures (beyond the RNN-based) and preprocessing methods, to improve our performance further.  It will also consider the questions of latency and complexity.  

We hope that research in these lines will produce powerful algorithms and have a substantial impact on error correction coding, rivaling the dramatic effect of deep learning on other fields.

\appendices
\section{Proof of Theorem~\ref{thrm:Main}}\label{apdx:proof}

In this section we provide the rigorous details of the proof, whose outline was provided in Sec~\ref{sec:analysis}.  

We begin with the following lemma.
\begin{lemma}\label{lemma:1}
	The following two claims hold with respect to the framework of Sec.~\ref{sec:1}:
	\begin{enumerate}
		\item There exists a matrix $A$ with dimensions $K\times N$, such that $A\bx^b = \bm$ for all $\bm$ and $\bx^b$ defined as in Sec.~\ref{sec:1} (recall that $\bm$ and $\bx^b$ are both column vectors).
		\item Let $B$ be a matrix obtained by concatenating the rows of $H$ and $A$ (i.e, $B=[H^T,A^T]^T)$.  	Then $B$ has full column rank, and is thus injective (one-to-one)
	\end{enumerate}
\end{lemma}

Note that in this lemma, we allow $H$ to contain redundant, linear dependent rows, as long as its rank remains $N-K$.  While we have not used such matrices in our work, the extra redundant rows could in theory be helpful in the design of effective neural networks for $\bF$.
\begin{proof}[Proof of Lemma~\ref{lemma:1}]
	Part~1 of the lemma follows simply from the properties of the generator matrix of the code $C$, denoted here $G$.  This matrix has full column-rank and dimensions $N\times K$ and satisfies $\bx^b = G\bm$.  Thus, we can define $A$ to be its left-inverse, and Part~1 of the lemma follows.\par
	To prove Part~2 of the lemma, we first observe that we can assume without loss of generality the matrix $H$ has full row-rank.   This is because by removing redundant (linear dependent) rows from $H$ we can obtain a full-rank parity-check matrix, and such removal cannot affect the rank of the corresponding $B$.  Let $D$ be a right-inverse of the matrix $H$, whose existence follows from the fact that $H$ has full row-rank.  $D$ has dimensions $N\times N-K$.  Consider the matrix $[G,D]$ (obtained by concatenating the columns of $G$ and $D$).  
	\begin{eqnarray*}
		B\cdot [G,D] &=& \begin{bmatrix}
			H \\
			A \\
		\end{bmatrix} \cdot \begin{bmatrix}
			G & D \\
		\end{bmatrix} = \begin{bmatrix}
			HG &  HD\\
			AG &  AD \end{bmatrix} \\&=& \begin{bmatrix}
			0 &  I_{N-K}\\
			I_{K} &  AD 
		\end{bmatrix},
	\end{eqnarray*}
	where $I_K$ and $I_{N-K}$ denote the identity matrices of dimensions $K$ and $N-K$, respectively.  The equality $HG = 0$ follows from the orthogonality of the generator and parity matrices, and equalities $HD = I_{N-K}$, $AG = I_K$ follow from the definitions above of $A$ and $D$.  The resulting matrix has rank $N$, and thus $B$ cannot have rank less than $N$.
\end{proof}

We now proceed to prove Part~1 of the theorem.  We use the following notation:  Vector are denoted by boldface (e.g., $\bx$) and scalars by normalface (e.g., $x$).  Random variables are upper-cased ($X$) and their instantiations lower-cased ($x$).  

We use the notation of Fig.~\ref{fig:1}, replacing lowercase with uppercase wherever we need to denote random variables.  Accordingly, we let $\bX$ and $\bY$ denote random variables corresponding to the channel input and output (respectively).  $\by$ is the realized channel output observed at the decoder and $x\in\{\pm1\}$ an arbitrary value.\par
	\begin{figure*}
		\begin{eqnarray}\label{eq:4}
		\Pr[X_i=x\given\bY=\by] &\addabove{=}{a}& \Pr[X_i\cdot Y^s_i =xy^s_i\:\given\:\bY=\by]\nonumber\\
		&\addabove{=}{b}& \Pr[\tZ^s_i = x y^s_i\:\given\: \bY=\by]\nonumber\\
		&\addabove{=}{c}& \Pr[\tZ^s_i = x y^s_i\:\given\:|\bY|=|\by|, \bY^s=\by^s]\nonumber\\
		&\addabove{=}{d}& \Pr[\tZ^s_i = x y^s_i\:\given\:|\btZ|=|\by|, \bY^b=\by^b]\nonumber\\
		&\addabove{=}{e}& \Pr[\tZ^s_i = x y^s_i\:\given\:|\btZ|=|\by|, B\bY^b=B\by^b]\nonumber\\
		&\addabove{=}{f}& \Pr[\tZ^s_i = x y^s_i\:\given\:|\btZ|=|\by|, H\bY^b=H\by^b,A\bY^b=A\by^b]\nonumber\\
		&\addabove{=}{g}& \Pr[\tZ^s_i = x y^s_i\:\given\:|\btZ|=|\by|, H[\bX^b\xor\btZ^b]=H\by^b,A[\bX^b\xor\btZ^b]=A\by^b]\nonumber\\
		&\addabove{=}{h}& \Pr[\tZ^s_i = x y^s_i\:\given\:|\btZ|=|\by|, H\btZ^b=H\by^b,\:\bM\xor A\btZ^b=A\by^b]\nonumber\\
		&\addabove{=}{i}& \Pr[\tZ^s_i = x y^s_i\:\given\:|\btZ|=|\by|, H\btZ^b=H\by^b]
		\end{eqnarray}
		\hrulefill
	\end{figure*}
	Consider the string of equations ending in~\eqref{eq:4}.  In (a), $Y^s_i$ is a random variable defined as $y^s_i$ (Fig.~\ref{fig:1}) and equality $Y^s_i=y^s_i$ follows from the condition $\bY=\by$.  In (b), we have relied on~\eqref{eq:5} to replace $\tZ^s_i=X_i\cdot Y^s_i$, where $\tZ^s_i\defined\sign(\tZ_i)$. In (c),  $|\bY|$ is the absolute value of $\bY$ and $\bY^s \defined \sign(\bY)$.  In (d), we have relied on~\eqref{eq:5} and the bipolarity of $\bX$ to obtain $|\btZ|=|\bY|$.  We have also defined and $\bY^b$ and $\by^b$ as in Fig.~\ref{fig:1}.  In (e), the matrix $B$ was defined as in Lemma~\ref{lemma:1} and equality follows by the fact that $B$ is injective.  In ($f$), we have used the definition $B=[H^T,A^T]^T$ where $A$ is as defined in Lemma~\ref{lemma:1}. In (g), we have decomposed $\bY^b=\bX^b\xor\btZ^b$.  This follows from~\eqref{eq:3} and~\eqref{eq:5}.  In (h), we have relied on $H\bX^b=\bzr$ which follows from the fact that $\bX^b$ is a valid codeword.  We have also replaced $A\bX^b=\bM$, where $\bM$ is the random message (see Fig.~\ref{fig:1}), following Lemma~\ref{lemma:1}.  Finally, in (i) we have made the observation that $\left[\bM\xor A\btZ^b\right]$ is independent of $\btZ$ and can therefore be omitted from the condition.  This follows because $\bM$, being the transmitted message, is uniformly distributed in $\{0,1\}^K$ (see Sec.~\ref{sec:1}) and independent of $A\btZ^b$ (and $\btZ$).
	
	The proof now follows from~(\ref{eq:4}).  To obtain the MAP decision for $X_i$ given $\by$ we can define the components of $\bF$ as follows, for $\bs\in\{0,1\}^{N-K}$ and $\ba \in \RR_+^N$.
	$$
	F_i(\ba, \bs) \defined \argmax_{\tz\in\{\pm 1\}} \left\{ \Pr[\tZ^s_i = \tz\:\given\:|\btZ|=\ba, H\btZ^b=\bs]\right\}
	$$
	By~(\ref{eq:4}), we now have $\MAP(X_i) = y^s_i F_i(|\by|, H\by^b)$.  Similarly, to obtain the MMSE estimate for $X_i$, we can define:
	$$
	F_i(\ba, \bs) \defined  \EE[\tZ^s_i\:\given\:|\btZ|=\ba, H\btZ^b=\bs]
	$$
	Recall from Fig.~\ref{fig:1} that in this case the decoder is configured to produce soft decisions and so the final sign operation is omitted (see Fig.~\ref{fig:1}).  This concludes the proof of Part~\ref{item:1} of the theorem.  
	
	Turning to Part~\ref{item:2}, in Sec~\ref{sec:analysis} we proved that the decoder's output can always be modeled as ${\bf f}(\btz)\cdot\bx$  for some vector-valued function ${\bf f} (\cdot)$.  With respect to the BER metric, the indices where the vectors ${\bf f}(\btz)\cdot\bx$ and $\bx$ diverge coincide with the indices where ${\bf f}(\btz)$ equals -1, and thus are independent of $\bx$.  With respect to the MSE metric (recall that in this case, the sign operation in Fig.~\ref{fig:1} is omitted), we have  $\MSE=\norm{{\bf f}(\btz)\cdot\bx - \bx}^2 = \norm{{\bf f}(\btz) - \bone}^2$ and thus the error is independent of $\bx$.
\QED

\section{Automorphisms of BCH codes}\label{apdx:BCH_Perm}
In this section, we discuss automorphisms of primitive narrow-sense binary BCH codes codes~\cite{Sloane}, which include the codes used in Sec~\ref{sec:simulation_results}.  

With the above codes, the blocklength $N$ equals $2^m-1$ for some positive integer $m$.  Codewords are binary vectors $\bc=[c_0,...,c_{N-1}]$ (i.e., defined over indices $i = 0,...,N-1$).  Permutations are bijective functions $\pi: \{0,...,N-1\} \rightarrow \{0,...,N-1\}$.  Given a codeword $\bc$ and a permutation $\pi$, we define the corresponding permuted codeword $\bc^\pi$  by $c^\pi_i = c_{\pi(i)}$.

The automorphism group of the above-mentioned BCH codes includes~\cite{Sloane} [pp.~233] permutations of the form:
\begin{eqnarray*}
	\pi_{k,l}(i) = 2^k i+ l \mod N
\end{eqnarray*}
where $k\in\{0,...m-1\}$ and $l\in\{0,...,N-1\}$.  The inverse permutation $\pi^{-1}_{k,l}$ can be shown to equals $\pi_{s,t}$ shere $s=m-k \mod m$ and $t = -2^s \cdot l \mod  N$.

We now address the question of efficiently finding the optimal permutation $\pi_{k,l}$ in the sense of Sec.~\ref{sec:preproc}, i.e., the permutation that maximizes the sum of adjusted reliabilities (see~\eqref{eq:R}) over the first $K$ components of the permuted codeword.  For fixed $k=0$, the set of permutations $\pi_{0,l}(i), l=0,...,N-1$ coincides with the set of cyclic permutations.  Finding the optimal cyclic permutation can efficiently be achieved by computing the cumulative sum of the adjusted reliabilities.  The case of arbitrary $k$ is adressed by observing that $\pi_{k,l}(i) = \pi_{0,l}(\pi_{k,0}(i))$.   The optimal permutation can be achieved by first applying the permutation $\pi_{k,0}$ and then repeating the above procedure for cyclic codes.  Finally, the optimal permutation across all $k$ is computed by combining the above results for each individual $k=0,...,m-1$.  With respect to computation latency, we note that the cumulative sum can be computed in logarithmic time by recursively splitting the range $0,...,N-1$.

Strictly speaking, the decoder of Fig.~\ref{fig:7} violates the framework of Sec.~\ref{sec:Theorem_Main}, because the preprocessing and postprocessing steps are not included in that framework.  In the case of BCH codes, however, this formal obstacle is easily overcome by removing the two steps and redefining $\bF$ to compensate.  Specifically, the preprocessing permutation can equivalently be realized within $\bF$ by permuting the vector $|\by|$ and manipulating the syndrome using identities detailed in~\cite{Sloane}.  The selection of the optimal permutation (which depends only on $|\by|$) and the postprocessing step can be redefined to be included in $\bF$.  The results of Sec.~\ref{sec:Theorem_Main} (particularly, resilience to overfitting) thus carry over to our setting.

\begin{small}

\end{small}


\begin{thebibliography}{99}


\bibitem{Urbanke_Message_Passing}
T.\ Richardson and R.\ Urbanke, ``The capacity of
low-density parity-check codes under message-passing decoding,'' {\it IEEE
Trans. Inf. Theory}, vol.~47, pp. 599--618, Feb. 2001.

\bibitem{Nachmani_1}
E.\ Nachmani, E.\ Marciano, L.\ Lugosch, Loren, W.J.\ Gross, D.\ Burshtein and Y.\ Be'ery, `` Deep learning methods for improved decoding of linear codes,'' {\it arXiv:1706.07043}, 2017


\bibitem{Nachmani_3}
E.\ Nachmani, Y.\ Bachar, E.\ Marciano, D.\ Burshtein and Y.\ Be’ery
``Near Maximum Likelihood Decoding with Deep Learning,'' {\it Int.  Zurich Seminar
	on Inf. and Comm.}, 2018

\bibitem{Hoydis_1}
T. J. O'Shea and J. Hoydis, ``An introduction to machine learning communications systems,''  {\it arXiv:1702.00832}, 2017.

\bibitem{Hoydis_2}
T. Gruber, S. Cammerer, J. Hoydis, and S. t. Brink, ``On deep learning-based channel decoding,'' {\it 51st Annual Conference on Inf. Sciences and Systems (CISS)}, 2017. 


\bibitem{Old_1}
 L.\ G.\ Tallini and P.\ Cull, ``Neural nets for decoding error-correcting codes,'' {\it Proc. IEEE Tech. Applicat. Conf. and Workshops Northcon95}, pp.\ 89--94, Oct.\ 1995. 

 
 \bibitem{Old_4}
 J.-L.\ Wu, Y.-H.\ Tseng, and Y.-M.\ Huang, ``Neural network decoders for linear block codes,'' {\it Int. Journ. of Computational Engineering Science,} vol. 3, no. 3, pp. 235--255, 2002.

\bibitem{Old_6}
J.\ Bruck and M.\ Blaum, ``Neural networks, error-correcting codes, and polynomials over the binary n-cube.'' {\it IEEE Trans. Inf. Theory} vol.\ 35(5), pp. 976--987, 1989.

\bibitem{OSD_1}
M.P.\ Fossorier,  S.\ Lin and J.\ Snyders. ``Reliability-based syndrome decoding of linear block codes.'' {\it IEEE Trans. Inf. Theory} vol.\ 44(1) pp. 388--398, Jan.\ 1998.

\bibitem{OSD_2}
M.P.\ Fossorier and  S.\ Lin, ``Soft-decision decoding of linear block codes based on ordered statistics.'' {\it IEEE Trans. Inf. Theory}, vol. 41(5), 1379--1396, Sep.\ 1995.



\bibitem{UrbankeBook}
T. Richardson and R.\ Urbanke, ``Modern coding theory,'' {\it Cambridge university press.} (2008)

\bibitem{CostelloBook}
S.\ Lin and D. J. Costello. ``Error control coding,'' 2nd edition, Prentice Hall, 2004.

\bibitem{Sason_1}
G.\ Wiechman and I.\ Sason, I., ``Parity-check density versus performance of binary linear block codes: New bounds and applications,''  {\it IEEE Trans. Inf. Theory} vol. 53(2), pp. 550--579, Jan. 2007.

\bibitem{Sason_2}
D.\ Burshtein, M.\ Krivelevich, S.\ Litsyn, S. and G.\ Miller, ``Upper bounds on the rate of LDPC codes,''.  {\it IEEE Trans. Inf. Theory} vol. 48(9), pp. 2437--2449, Sep. 2002.

\bibitem{DeepLearningTextBook}
I.\ Goodfellow, Y.\ Bengio, and A. Courville, ``Deep learning.'' MIT press, 2016.

\bibitem{Adam}
D.\ Kingma and J.\ Ba, ``Adam: A method for stochastic optimization.'' {\it arXiv preprint arXiv:1412.6980.} Dec. 2014.

\bibitem{Beery}
I.\ Dimnik and Y.\ Be'ery, ``Improved random redundant iterative HDPC decoding.'' {\it IEEE Trans. Commu.}, 57(7), July 2009.

\bibitem{Sloane}
F.\ J.\ MacWilliams and N.\ J.\ A.\ Sloane, {\it The Theory of Error-Correcting
Codes.} North-Holland, 1978.

 
 
\bibitem{graves2013speech}
A.\ Graves, M.\ Abdel-rahman and G.\ Hinton, ``Speech recognition with deep recurrent neural networks'' {\it IEEE international conference on acoustics, speech and signal processing (icassp)} 2013.
 
\bibitem{cho2014learning}
K.\ Cho, B.\ Van Merri{\"e}nboer, C.\ Gulcehre , D.\ Bahdanau, F.\ Bougares ,H.\ Schwenk and Y.\ Bengio G.\ Hinton, ``Learning phrase representations using RNN encoder-decoder for statistical machine translation'' {\it arXiv preprint arXiv:1406.1078} 2014.
 
 


\bibitem{hochreiter1997long}
S.\ Hochreiter and J.\ Schmidhuber, ``Long short-term memory'' {\it Neural computation} 1997.




\end{thebibliography}
\end{document}